\documentclass[a4paper,11pt]{article}

\pdfoutput=1
\usepackage{amsthm} 
\usepackage{colortbl}
\usepackage[fleqn]{amsmath}
\usepackage[psamsfonts]{amssymb}
\usepackage{stmaryrd, enumerate, mathrsfs, mathcomp}
\usepackage[linesnumbered,ruled]{algorithm2e}

\newtheorem{theorem}{Theorem}

\newtheorem{conjecture}{Conjecture}
\newtheorem{problem}{Problem}

\newcommand{\COMPLEXITYCLASSNAMESTYLE}[1]{\mathsf{#1}}

\newcommand{\PROBLEMNAMESTYLE}[1]{\textsc{#1}}
\newcommand{\PROBLEMNAMEDELIMITER}{\hspace{0.5mm}}

\newcommand{\AIKP}{\XIKP{Assoc}}
\newcommand{\CLASSP}{\COMPLEXITYCLASSNAMESTYLE{P}}
\newcommand{\COFROB}{\PROBLEMNAMESTYLE{co\FROB}}
\newcommand{\CONP}{\COMPLEXITYCLASSNAMESTYLE{co}\NP}

\newcommand{\DELTAP}[1]{\COMPLEXITYCLASSNAMESTYLE{\Delta}_{#1}^\CLASSP}
\newcommand{\DIFFP}[1]{\COMPLEXITYCLASSNAMESTYLE{D}_{#1}^\CLASSP}
\newcommand{\DP}{\COMPLEXITYCLASSNAMESTYLE{DP}}
\newcommand{\EFROB}{\EXACT\PROBLEMNAMEDELIMITER\FROB}
\newcommand{\EXACT}{\PROBLEMNAMESTYLE{Exact}}
\newcommand{\FFROB}{\PROBLEMNAMESTYLE{F}\FROB}
\newcommand{\FP}{\FUNC\CLASSP}
\newcommand{\FROB}{\PROBLEMNAMESTYLE{Frobenius}}
\newcommand{\FROBCOFROB}{\FROB\text{-}\COFROB}

\newcommand{\FUNC}{\COMPLEXITYCLASSNAMESTYLE{F}}
\newcommand{\IKP}{\PROBLEMNAMESTYLE{Integer\PROBLEMNAMEDELIMITER}\KNAPSACK}
\newcommand{\INCOMPLETEGAME}{\PROBLEMNAMESTYLE{Incomplete{\PROBLEMNAMEDELIMITER}Game}}

\newcommand{\KNAPSACK}{\PROBLEMNAMESTYLE{Knapsack}}

\newcommand{\MATCH}{\PROBLEMNAMESTYLE{Matching}}
\newcommand{\NP}{\COMPLEXITYCLASSNAMESTYLE{NP}}

\newcommand{\PI}[1]{\COMPLEXITYCLASSNAMESTYLE{\Pi}_{#1}}
\newcommand{\PIP}[1]{\COMPLEXITYCLASSNAMESTYLE{\Pi}_{#1}^\CLASSP}
\newcommand{\PLS}{\COMPLEXITYCLASSNAMESTYLE{PLS}}

\newcommand{\PPAD}{\COMPLEXITYCLASSNAMESTYLE{PPAD}}

\newcommand{\SAT}{\PROBLEMNAMESTYLE{Sat}}
\newcommand{\SIGMA}[1]{\COMPLEXITYCLASSNAMESTYLE{\Sigma}_{#1}}
\newcommand{\SIGMAMM}[1]{\COMPLEXITYCLASSNAMESTYLE{\Sigma}_{#1}^\COMPLEXITYCLASSNAMESTYLE{MM}}
\newcommand{\SIGMAP}[1]{\COMPLEXITYCLASSNAMESTYLE{\Sigma}_{#1}^\CLASSP}
\newcommand{\SIGMAPISAT}[1]{\SIGMA{#1}\SAT\text{-}\PI{#1}\SAT}
\newcommand{\XIKP}[1]{\PROBLEMNAMESTYLE{#1\PROBLEMNAMEDELIMITER}\IKP}
\newcommand{\XMATCH}[1]{\PROBLEMNAMESTYLE{#1}\PROBLEMNAMEDELIMITER\MATCH}

\definecolor{gray}{rgb}{0.86,0.86,0.86}

\begin{document}

\title{The Computational Complexity \\
of the Frobenius Problem}
\author{%
Shunichi Matsubara\\
Aoyama Gakuin University\\
{\ttfamily matsubara@it.aoyama.ac.jp}
}
\date{}

\maketitle

\begin{abstract}
In this paper, as a main theorem, 
we prove that the decision version of the Frobenius problem is 
$\SIGMAP{2}$-complete under Karp reductions.
Given a finite set $A$ of coprime positive integers,
we call the greatest integer that
cannot be represented as a nonnegative integer combination of $A$
the {\itshape Frobenius number}, and we denote it as $g(A)$.
We call a problem of finding $g(A)$ for a given $A$ the {\itshape Frobenius problem};
moreover, we call a problem of determining whether $g(A) \geq k$
for a given pair $(A, k)$ the decision version of the Frobenius problem,
where $A$ is a finite set of coprime positive integers
and $k$ is a positive integer.
For the proof, 
we construct two Karp reductions.
First, we reduce
a $2$-alternating version of the $3$-dimensional matching problem,
which is known to be $\PIP{2}$-complete,
to a $2$-alternating version of the integer knapsack problem. 
Then, we reduce the variant of the integer knapsack problem
to the complement of the decision version of the Frobenius problem.
As a corollary, we obtain the main theorem.
\end{abstract}

\section{Introduction}
\label{sec:introduction}

The Frobenius problem has attracted the interest of a number of mathematicians and computer scientists
since the $19$-th century~(\cite{Sylvester1882AmJMath5},
Problem C7 in \cite{guy2003unsolved},
\cite{ramirezAlfonsin2005diophantine},
and Chapter 1 in \cite{BeckRobins2015}).
Let $A = \{a_1,\cdots,a_n\}$ be a set of coprime integers 
such that $2 \leq a_1 < \cdots < a_n$,
where $n \geq 2$.
We call the greatest integer
that cannot be represented as a nonnegative integer combination of $A$
the {\itshape Frobenius number} of $A$, 
and we denote it as $g(A)$.
For example, given $\{4,6,7\}$,
the Frobenius number $g(\{4,6,7\})$ is $9$.
Generally,
a function problem that asks for the Frobenius number 
for a given finite set of coprime positive integers is called the 
{\itshape Frobenius problem}~\cite{ramirezAlfonsin2005diophantine}.
In this paper,
we denote this original version of the problem as $\FFROB$.
We denote the decision version of the Frobenius problem as $\FROB$,
which determines whether $g(A) \geq k$ for a given finite set $A$ 
of coprime positive integers and a positive integer $k$.
Moreover, we denote the complement problem of $\FROB$ as $\COFROB$.

\subsection{Results of This Work}
\label{subsec:1_1}

In this paper, 
we prove that $\FROB$ is $\SIGMAP{2}$-complete under Karp reductions
as a main theorem.
This result provides the first nontrivial upper bound 
and an improved lower bound for the computational complexity of $\FFROB$.
$\FFROB$ has been proven to be $\NP$-hard
under Cook reductions~\cite{ramirezAlfonsin1996ComplexityOfFP}.
However, to the best of the author's knowledge,
little other research has been conducted 
on any variant of the Frobenius problem
from a complexity theoretical perspective.
$\SIGMAP{2}$ is the complexity class at the second level of
the polynomial hierarchy~\cite{MeyerStockmeyer1972}.
Every problem in $\SIGMAP{2}$ can be computed in 
nondeterministic polynomial time by using
an $\NP$ oracle.
$\PIP{2}$ is the class of the complements of problems in $\SIGMAP{2}$.
Ram\'irez-Alfons\'in proposed an open question of whether
$\FROB$ is $\NP$-complete under Karp reductions
in his monograph~(Section A.1 in \cite{ramirezAlfonsin2005diophantine}).
This work is also an answer for that open question.

We prove the $\SIGMAP{2}$-completeness of $\FROB$ as follows.
First, we construct a Karp reduction
from $\PI{2}\XMATCH{3D}$ to $\PI{2}\AIKP$.
$\PI{2}\XMATCH{3D}$ is a $2$-alternating variant of the $3$-dimensional matching
problem~(Section A3.2 in \cite{garey1979computers}).
$\PI{2}\XMATCH{3D}$ is known to be $\PIP{2}$-complete due to \cite{1056978}.
$\PI{2}\AIKP$ is a $2$-alternating version 
of the integer knapsack problem~(Section A6 in
\cite{garey1979computers} and Section 15.7 in \cite{Papadimitriou:1982:COA:31027}).
$\PI{2}\AIKP$ is introduced in this paper.
We define this problem by associating with $\FROB$.
Then, we prove the membership of $\COFROB$ to $\PIP{2}$
and the $\PIP{2}$-hardness of $\COFROB$.
The $\PIP{2}$-hardness is proven 
by constructing a Karp reduction from $\PI{2}\AIKP$ to $\COFROB$.
As a corollary, we obtain the $\SIGMAP{2}$-completeness of $\FROB$.
This means that $\SIGMAP{2}$ is a lower bound for the complexity class of $\FFROB$.

Moreover, as a corollary of the $\SIGMAP{2}$-completeness of $\FROB$,
we show that $\FUNC\DELTAP{3}$ is an upper bound for $\FFROB$.
Then, we demonstrate that the $\SIGMAP{2}$-hardnesses of $\FROB$ and $\FFROB$ are weak
in the sense that there are pseudopolynomial algorithms.

\subsection{Related Work}
\label{subsec:1_2}

\subsubsection{Fast Algorithms for Solving the Frobenius Problem}
\label{subsubsec:1_2_1}

Prior to this work,
the computational difficulty of the Frobenius problem
was recognized based on the result of \cite{ramirezAlfonsin1996ComplexityOfFP}
from a theoretical perspective.
However, 
many practically fast algorithms have been actively developed~(Chapter 1 in \cite{ramirezAlfonsin2005diophantine}).
Nijenhuis~\cite{nijenhuis1979} developed a practically fast algorithm
for any instance.
He provided a characterization of the Frobenius number
by a weighted directed graph.
His algorithm is a variant of Dijkstra's algorithm for 
the single-source shortest path problem over the graph.
B\"{o}cker and Lipt\'ak~\cite{BockerSebastiAndLiptakZsuzsanna2007}
developed a practically fast algorithm
and applied it to solve a problem in bioinformatics.
Einstein, Lichtblau, Strzebonski, and Wagon~\cite{einstein2007frobenius}
developed some algorithms through the use of some mathematical
programming techniques.
For a given set $A$ of input integers,
if the number of elements of $A$ is of logarithmic order of the smallest element of $A$,
then that algorithm can run very fast.

Beihoffer, Hendry, Nijenhuis, and Wagon~\cite{beihoffer2005faster}
developed some algorithms 
by extending the algorithm of Nijenhuis~\cite{nijenhuis1979}.
Practically, their algorithms are considered to be the best algorithms under no restriction for inputs.
Roune~\cite{Roune20081} implemented an algorithm using Gr\"obner bases.
Using this algorithm, he computed the Frobenius numbers 
for inputs of thousands digits.
However, this algorithm is only practical 
if the number of a given set of positive integers is sufficiently small.

\subsubsection{Computation of the Frobenius Problem for Inputs of a Fixed Number of Integers}
\label{subsubsec:1_2_2}

If we assume that the number $n$ of input positive integers 
is fixed, then there are polynomial-time algorithms.
In the case where $n = 2$,
for any coprime positive integers $a_1$ and $a_2$,  
the Frobenius number can be calculated using
the formula $a_1 a_2 - a_1 - a_2$,
whose discoverer is unknown.
In the case where $n = 3$,
polynomial-time algorithms 
are known, e.g., the algorithm 
proposed by Davison~\cite{Davison1994353}.
In the case of any $n \geq 2$,
polynomial-time algorithms were found
by Kannan~\cite{Kannan1992} and Barvinok and Woods~\cite{barvinok2003short}.

\subsubsection{Upper and Lower Bounds for the Frobenius Number}
\label{subsubsec:1_2_3}

Although $\FROB$ cannot be efficiently computed 
unless $\CLASSP=\NP$,
some upper bounds are known for the Frobenius number.
Let $A = \{a_1,\cdots,a_n\}$ be 
a set of coprime integers
such that $2 \leq a_1 < \cdots < a_n$,
where $n \geq 2$.
For example, the following 
general upper bounds are known.
A simple upper bound $a_n^2$ was found by 
Wilf~\cite{Wilf1978}.
Another upper bound
$2 a_n \lfloor a_1/n \rfloor - a_1$
was found by Erd\"{o}s and Graham~\cite{erdos1972linear}.
The upper bound found by Krawczyk and Paz~\cite{Krawczyk1989289} 
is attractive because this bound has the same order of magnitude as the Frobenius number
and can be computed in polynomial time
under the assumption that $n$ is fixed.
Generally, 
we cannot know whether a bound is superior to another
since it depends on a given set of positive integers. 
Some lower bounds are also known.
For example,
Davison~\cite{Davison1994353} found a sharp lower bound
$\sqrt{3 a_1 a_2 a_3} - a_1 - a_2 - a_3$
of the Frobenius number
in the case where $n = 3$.
Aliev and Gruber~\cite{Aliev200771} found that
$\left(\left(n-1\right)!\Pi_{i=1}^n a_i\right)^{1/(n-1)} -
\sum_{i=1}^n a_i$. 

\subsubsection{Computational Complexity}
\label{subsubsec:1_2_4}

The complexity classes $\SIGMAP{2}$ and $\PIP{2}$ 
have been actively researched.
Stockmeyer proved 
that the problems $\SIGMA{2}\SAT$ and $\PI{2}\SAT$ 
are $\SIGMAP{2}$-complete and $\PIP{2}$-complete,
respectively~\cite{Stockmeyer19761}.
$\SIGMA{2}\SAT$ and $\PI{2}\SAT$
are extensions to the
$\SIGMAP{2}$ and $\PIP{2}$ variants
of the satisfiability problem, respectively.
McLoughlin proved that the covering radius problem
for linear codes is  $\PIP{2}$-complete~\cite{1056978}.
This problem is defined as follows.
{\itshape Given a pair $(A,w)$,
where $A$ is an $(m,n)$-matrix and $w$ is an integer,
for any $n$-vector $y$, 
is there an $m$-vector $x$ such that
$x A = y$ and 
the Hamming weight of $x$ is not greater than $w$?}
She showed the $\PIP{2}$-completeness  
of the covering radius problem
using two Karp reductions.
First, she reduced $\PI{2}\SAT$ to $\PI{2}\XMATCH{3D}$;
then, she reduced $\PI{2}\XMATCH{3D}$ to 
the covering radius problem.
Umans~\cite{Christopher2001597}
proved that the minimum equivalent DNF problem
is $\SIGMAP{2}$-complete.
This problem is defined as follows.
{\itshape Given a pair $(\varphi, k)$,
where $\varphi$ is a Boolean formula
and $k$ is an integer,
is there an equivalent formula $\psi$ to $\varphi$
with at most $k$ occurrences of literals?}
Umans showed the $\SIGMAP{2}$-completeness 
of the minimum equivalent DNF problem
using two Karp reductions.
A survey by 
Sch{\"a}fer and Umans~\cite{schaefer2002completeness}
provided a comprehensive list of
numerous problems 
at the second and third levels in the polynomial hierarchy
and their related results,
which was written in the style of \cite{garey1979computers}.

\subsubsection{Covering Radius Problem}
\label{subsubsec:1_2_5}

The Frobenius problem
is closely related to the covering radius problem 
for lattices and linear codes.
The result of McLoughlin~\cite{1056978}, described above,
is an example. 
The covering radius problem and the Frobenius problem
belong to classes at the second level of the polynomial hierarchy,
although the membership of the Frobenius problem will be shown 
in a later section of this paper.
As more general research for the complexity of the covering radius problem,
Guruswami, Micciancio, and Regev~\cite{GuruswamiMicciancioRegev2005}
investigated the approximability
of the covering radius problem and its related problems
for lattices and linear codes.
Kannan~\cite{Kannan1992} found the following relation 
for the Frobenius number and a type of covering radius for a lattice.
Given a set $A$ of coprime integers $a_1,\cdots,a_n$ such that
$2 \leq a_1 < \cdots < a_n$,
$R(P,L)$ is equal to $g(A) + \sum_{i=1}^{n} a_i$.
$P$ is a polytope such that
$(x_1,\cdots,x_{n-1}) \in P$ if and only if  
$x_1,\cdots,x_{n-1}$ are real numbers and
$\sum_{i=1}^{n-1} a_i x_i \leq 1$.
$L$ is a lattice such that
$(x_1,\cdots,x_{n-1}) \in L$ if and only if
$x_1,\cdots,x_{n-1}$ are integers and
$\sum_{i=1}^{n-1} a_i x_i$ is congruent to $0$ modulo $a_n$.
$R(P,L)$ is a covering radius of $P$ for $L$.

\subsection{Organization}

The remainder of this paper is organized as follows.
In Section~\ref{sec:preliminaries},
we define some related concepts and notations.
In Section~\ref{sec:completeness_PiAIK},
we prove the $\PIP{2}$-completeness of $\PI{2}\AIKP$.
Then, in Section~\ref{sec:completeness_FROB},
we prove $\FROB$ to be $\SIGMAP{2}$-complete as a main theorem.
In Section~\ref{sec:completeness_FFROB}, 
we discuss lower bounds and upper bounds for $\FFROB$ by using the main theorem.
Section~\ref{sec:strength} describes the weak $\SIGMAP{2}$-hardnesses
of $\FFROB$ and $\FROB$.
Finally, Section~\ref{sec:conclusion} concludes this work 
and describes open problems and future work.

\section{Preliminaries}
\label{sec:preliminaries}

\subsection{Basic Concepts and Notations}

We denote the sets of all nonnegative and positive integers
as $\mathbb{N}$ and $\mathbb{N}_+$, respectively.
For any $i, j$ in $\mathbb{N}$ with $i \leq j$,
we denote the integer interval
$\{k \in \mathbb{N} \colon i \leq k \leq j\}$ 
as $[i,j]$.

\subsection{Representations for Positive Integers}

For any $n \in \mathbb{N}$,
let $\overline{n}$ be a new symbol.
For any $N \subseteq \mathbb{N}$,
let $\overline{N}$ denote the set
$\{\overline{n} \colon n \in \mathbb{N}\}$.
Let $b$ be a nonnegative integer.
Let $n_1, \cdots, n_k$ be integers in $[0,b-1]$.
Then, we call the sequence 
$\overline{n_k} \, \cdots \overline{n_1}$
the {\itshape $k$-place $b$-representation} for integer
$\sum_{i=1}^{k} n_i b^{i-1}$.
We often denote the integer $\sum_{i=1}^{k} n_i b^{i-1}$
as $(\overline{n_k} \, \cdots \overline{n_1})_b$.
For notational convenience, 
we often denote a $k$-place $b$-representation
$\overline{n_k} \, \cdots \overline{n_1}$
as $\overline{n_k \cdots n_1}$.
For any $k$-place $b$-representation $\overline{n_k} \, \cdots \overline{n_1}$,
we call $k$ and $b$ its {\itshape length} and {\itshape base},
respectively.
Moreover, for every $i \in [1,k]$, 
we call $n_i$ its {\itshape $i$-th digit}.
We often omit {\itshape ``$k$-place''} or {\itshape ``$b$-''}.
Let $r$ be $\overline{n_k} \cdots \overline{n_1}$.
For any $i \in [1,k]$,
we denote the $i$-th digit $\overline{n_i}$ as $r[i]$.
For every $i,j \in [1,k]$ with $i \leq j$,
we call $\overline{n_i} \cdots \overline{n_j}$
a {\itshape subrepresentation} and denote it as $r[i,j]$.
For every $n \in [0,b-1]$ and $m \in \mathbb{N}$,
we define ${\overline{n}}^m$ inductively as follows.
(1) ${\overline{n}}^0 = \varepsilon$, 
(2) $\overline{n}^{m+1} = \overline{n}^{m} \overline{n}$ 
where $\varepsilon$ denotes the empty representation.

We apply some concepts on integers to their $b$-representations.
We define an {\itshape addition} of $b$-representations as follows.
Let $r_1,\cdots,r_l$ be 
$k$-place $b$-representations, 
where $b, k,l \in \mathbb{N}_+$.
Let $i$ be an integer in $[0,k]$.
Then, we define integers $d_i$ and $c_i$ 
inductively as follows.
(1) $d_0 = 0$ and $c_0 = 0$.
(2) If $i \in [1,k]$, then
$d_i$ is the floor of the quotient of $\sum_{j=1}^{l} {(r_j[i])}_b + c_{i-1}$
divided by $b$, and 
$c_i$ is the remainder of $\sum_{j=1}^{l} {(r_j[i])}_b + c_{i-1}$ 
divided by $b$. 
Then, we call $\overline{d_k} \cdots \overline{d_1}$
the {\itshape sum} of $r_1,\cdots,r_l$,
and the operation for computing the sum
is the {\itshape addition} of $r_1,\cdots,r_l$.
We call $c_i$ the {\itshape carry} at the $i$-th digit
in the addition.
We say that a carry {\itshape occurs} at the $i$-th digit
in the addition of $r_1,\cdots,r_l$
if $c_i \neq 0$.
We define ordering relations
$<, \leq,  =,  \geq,  >$
over $b$-representations as follows. 
Let $n$ and $n^\prime$ be nonnegative integers,
and let $r$ and $r^\prime$ be their $b$-representations, respectively.
Let $\circ$ be any symbol in $\{<, \leq, =, \geq, >\}$.
Then, $r \circ r^\prime$ if and only if $n \circ n^\prime$.

\subsection{Complexity Classes}
\label{subsec:complexity_classes}

In this subsection, we review some fundamental concepts 
that are closely related to this paper.
We assume that the reader is familiar with the basis of computational complexity theory.
If necessary, the reader is referred to some standard textbooks, 
e.g., \cite{arora2009computational,garey1979computers}.
We define the classes 
$\DELTAP{k}$, $\SIGMAP{k}$, and $\PIP{k}$, where $k \geq 0$, 
in the polynomial hierarchy inductively as follows.
Each of $\DELTAP{0}$, $\SIGMAP{0}$, and $\PIP{0}$ 
is the class $\CLASSP$.
For any $k \geq 1$, 
$\DELTAP{k}$, $\SIGMAP{k}$, and $\PIP{k}$ 
are the classes 
$\CLASSP^{\SIGMAP{k-1}}$, $\NP^{\SIGMAP{k-1}}$, and $\CONP^{\SIGMAP{k-1}}$,
respectively.
These definitions originate from \cite{MeyerStockmeyer1972}.
For every $k \geq 0$, we define $\DIFFP{k}$ as the class
of all of the problems $L$ such that 
$L$ is the intersection of some $L_1 \in \SIGMAP{k}$ and $L_2 \in \PIP{k}$.
This definition originates from \cite{WOOLDRIDGE200427}.
By definition,
the classes $\SIGMAP{1}$, $\PIP{1}$, and $\DIFFP{1}$ are identical to
$\NP$, $\CONP$, and $\DP$, respectively.

\subsection{Computational Problems}
\label{subsec:computational_problems_definitions}

In this subsection, we summarize the computational problems
described in this paper.
Given a problem $L$,
for every instance $I$ of $L$,
we define the size of $I$
as the bit length for representing $I$.

\begin{problem}[$\FFROB$]
\label{ex:1}
\

\noindent{\itshape Input:}
A set $A$ of coprime integers $a_1,\cdots,a_n$ 
such that $2 \leq a_1 < \cdots < a_n$ and $n \geq 2$.

\noindent{\itshape Output:}
$g(A)$.
\end{problem}

\begin{problem}[$\FROB$]
\label{prob:4}
\

\noindent{\itshape Instance:}
A pair $(A, k)$, where
$A$ is a set of coprime integers $a_1,\cdots,a_n$ 
such that $2 \leq a_1 < \cdots < a_n$ and $n \geq 2$,
and $k \in \mathbb{N}_+$.

\noindent{\itshape Question:}
$g(A) \geq k$?
\end{problem}

\begin{problem}[$\COFROB$]
\label{prob:5}
\

\noindent{\itshape Instance:}
A pair $(A, k)$, where
$A$ is a set of coprime integers $a_1,\cdots,a_n$ 
such that $2 \leq a_1 < \cdots < a_n$ and $n \geq 2$,
and $k \in \mathbb{N}_+$.

\noindent{\itshape Question:}
$g(A) < k$?
\end{problem}

\begin{problem}[$\EFROB$]

\noindent{\itshape Instance:}
A pair $(A, k)$, where
$A$ is a set of coprime integers $a_1,\cdots,a_n$ 
such that $2 \leq a_1 < \cdots < a_n$ and $n \geq 2$,
and $k \in \mathbb{N}_+$.

\noindent{\itshape Question:}
$g(A) = k$?

\end{problem}

\begin{problem}[$\FROBCOFROB$]

\noindent{\itshape Instance:}
A $4$-tuple $(A_1, k_1; A_2, k_2)$, where
$(A_1,k_1)$ and $(A_2,k_2)$ are
instances of $\FROB$ and $\COFROB$, respectively.

\noindent{\itshape Question:}
$g(A_1) \geq k_1$ and $g(A_2) < k_2$?

\end{problem}

\begin{problem}[$\PI{2}\XMATCH{3D}$]
\label{prob:1}
\

\noindent
{\itshape Instance:}
A $5$-tuple $(U_1,U_2,U_3,M_1,M_2)$, where
$U_1,U_2,U_3$ are disjoint sets such that
$|U_1| = |U_2| = |U_3| = q$ for some $q \in \mathbb{N}$,
and $M_1,M_2 \subseteq U_1 \times U_2 \times U_3$.

\noindent
{\itshape Question:}
For every $\mu_1 \subseteq M_1$,
is there $\mu_2 \subseteq M_2$ such that
$\mu_1 \cup \mu_2$ is not a matching?

\noindent{\itshape Comment:}
This problem was proven to be $\PIP{2}$-complete~\cite{1056978},
although she used the name ``AE $3$-dimensional matching''
rather than $\PI{2}\XMATCH{3D}$ in that paper.

\end{problem}

In this paper, we call a 3-dimensional matching
simply a matching if no confusion arises.
Moreover,
we define the following two total orders 
in an instance of $\PI{2}\XMATCH{3D}$,
which are specified by the subscripts.
For every $i \in [1,3]$, let
$u_{i,1},\cdots,u_{i,q}$
denote all elements of $U_i$.
We define a relation $<$ on $U_i$
as a total order such that $u_{i,1} < \cdots < u_{i,q}$.
We define a relation $<$ on $U_1 \times U_2 \times U_3$ 
as a total order such that
$(u_{1,j_1},u_{2,j_2}, u_{3,j_3})
< (u_{1,k_1},u_{2,k_2}, u_{3,k_3})$
if $(\overline{j_1} \ \overline{j_2} \ \overline{j_3})_{q+1} 
< (\overline{k_1} \ \overline{k_2} \ \overline{k_3})_{q+1}$
for every $j_1,j_2,j_3,k_1,k_2,k_3 \in [1,q]$.

\begin{problem}[$\IKP$]
\label{prob:10}
\

\noindent{\itshape Instance:}
A triple $A$, where
$A$ is a set of positive integers  $a_1,\cdots,a_n$
for some $n \in \mathbb{N}$.

\noindent{\itshape Question:}
Are there nonnegative integers $x_1,\cdots,x_n$
such that $\sum_{i=1}^n x_i a_i = k$?

\noindent{\itshape Comment:}
This problem was proven to be $\NP$-complete%
~(Section 15.7 in \cite{Papadimitriou:1982:COA:31027}).

\end{problem}

\begin{problem}[$\PI{2}\IKP$]
\label{prob:2}
\

\noindent{\itshape Instance:}
A triple $(A, \lambda, \upsilon)$, where
 $A$ is a set of positive integers
 $a_1,\cdots,a_n$
 for some $n \in \mathbb{N}$,
 and $\lambda, \upsilon \in \mathbb{N}_+$.

\noindent{\itshape Question:}
For every $k \in [\lambda, \upsilon]$, 
are there nonnegative integers $x_1,\cdots,x_n$
such that
$\sum_{i=1}^n x_i a_i = k$?
\end{problem}

Furthermore, we define a subproblem of $\PI{2}\IKP$,
which is associated with the Frobenius problem
as follows.

\begin{problem}[$\PI{2}\AIKP$]
\label{prob:3} \

\noindent{\itshape Instance:}
A pair $(A, \lambda)$,
where $(A, \lambda, \lambda+c)$ is an instance of $\PI{2}\IKP$
and $c = \min A - 1$.

\noindent{\itshape Question:}
Is $(A, \lambda, \lambda+c)$ a yes instance of $\PI{2}\IKP$? 
\end{problem}
We introduce $\PI{2}\AIKP$ 
only for the simulation of $\COFROB$.
Although the restriction of $\PI{2}\AIKP$ 
appears to be unnatural,
it suffices to argue the computational complexity of $\PI{2}\AIKP$
for proving the $\SIGMAP{2}$-completeness of $\FROB$.

\begin{problem}[$\SIGMA{k}\SAT$, $k \geq 1$]
\label{prob:6} \

\noindent{\itshape Instance:}
A CNF formula $\varphi$ over disjoint sets $X_1,\cdots,X_k$.

\noindent{\itshape Question:}
$(\exists \sigma_1 \in \{0,1\}^{|X_1|})
(\forall \sigma_2 \in \{0,1\}^{|X_1|}) \cdots (Q \sigma_k \in \{0,1\}^{|X_1|})[\varphi(\sigma_1 \cdots \sigma_k) = 1]$?
Here, $Q$ is the existential quantifier if $k$ is odd; otherwise, it
is the universal one.

\noindent{\itshape Comment:}
This problem was proven to be $\SIGMAP{2}$-complete~\cite{Wrathall197623}.

\end{problem}

\begin{problem}[$\PI{k}\SAT$, $k \geq 1$]
\label{prob:7} \

\noindent{\itshape Instance:}
A CNF formula $\varphi$ over disjoint sets $X_1,\cdots,X_k$.

\noindent{\itshape Question:}
$(\forall \sigma_1 \in \{0,1\}^{|X_1|})
(\exists \sigma_2 \in \{0,1\}^{|X_1|}) \cdots (Q \sigma_k \in \{0,1\}^{|X_1|})[\varphi(\sigma_1 \cdots \sigma_k) = 1]$?
Here, $Q$ is the universal quantifier if $k$ is odd; otherwise, it is the existential one.

\noindent{\itshape Comment:}
This problem was proven to be $\PIP{2}$-complete~\cite{Wrathall197623}.

\end{problem}

\begin{problem}[$\SIGMAPISAT{k}$, $k \geq 1$]
\label{prob:8} \

\noindent{\itshape Instance:}
A pair $(\varphi, \psi)$ of
CNF formulae over disjoint sets $X_1,\cdots,X_k$.

\noindent{\itshape Question:}
Are $\varphi$ and $\psi$ 
yes instances of $\SIGMA{k}\SAT$ and $\PIP{k}\SAT$, respectively?

\noindent{\itshape Comment:}
This problem was proven to be $\DIFFP{2}$-complete~\cite{WOOLDRIDGE200427}.

\end{problem}

\subsection{Other Measure for Analyzing Computational Complexity}
\label{subsec:strength_SIGMAP_definition}

In this subsection, we define strong $\NP$-hardness
and its related concepts~(\cite{Garey:1978:SNR:322077.322090}, Section 4.2 in \cite{garey1979computers}).
In this paper, 
we analyze the complexity of a problem
primarily by using the bit length of a given instance as only one parameter.
However, even if a problem $L$ does not have an algorithm that is polynomial time
in the bit length,
the problem can have a polynomial-time algorithm under other assumptions for its parameters.
Indeed, by measuring the complexity under a different assumption, 
we can refine many complexity classes.
This measure is also useful if we study the
approximability of computationally hard problems~(Chapter 6 in \cite{garey1979computers}, \cite{433837220000101}).

Given a problem $L$ and an instance $I$ of $L$,
if some components of $I$ are integers, 
then we call the maximum magnitude of such integers
the {\itshape unary size} of $I$.
For example,
for an instance $(A, k)$ of $\FROB$,
the unary size of $(A,k)$ is 
$\max (\{a \colon a \in A\} \cup \{k\})$.
For a problem $L$,
even if any instance $I$ of $L$ does not explicitly include
any integer component,
we can define the unary size of $I$
by considering the labels of its components to be reasonably encoded to integers.
For example, 
given an instance $\varphi$ of $\SAT$,
we define the unary size of $\varphi$
as the larger of the numbers of the clauses and variables.
If necessary,
we call the bit length of an input the {\itshape binary input} explicitly.

For any problem $L$,
we call an algorithm $A$ a {\itshape pseudopolynomial} algorithm for solving $L$
if $A$ can solve $L$ in polynomial time in the unary and binary sizes.
Let $\mathcal{C}$ be a complexity class.
Moreover, let $L$ be a $\mathcal{C}$-hard problem.
Then, we say that $L$ is {\itshape strongly} $\mathcal{C}$-hard
if there is no pseudopolynomial algorithm for solving $L$;
otherwise, we say that it is {\itshape weakly} $\mathcal{C}$-hard.

\section{$\PIP{2}$-Completeness of 
the Integer Knapsack Problem Associated with the Frobenius Problem}
\label{sec:completeness_PiAIK}

In this section, we prove
the $\PIP{2}$-completeness of $\PI{2}\AIKP$ under Karp reductions.
We construct a reduction from a $\PIP{2}$-complete problem $\PI{2}\XMATCH{3D}$ to $\PI{2}\AIKP$.
We first describe the key concepts of the reduction 
by using examples in Subsection~\ref{subsec:ideas_reduction_Pi3DM}. 
Then, we prove the $\PIP{2}$-completeness in Subsection~\ref{subsec:formulation_reduction_3DM}.

\subsection{Ideas of a Reduction from $\PI{2}\XMATCH{3D}$ to $\PI{2}\AIKP$}
\label{subsec:ideas_reduction_Pi3DM}

In this subsection, we describe the key ideas of our reduction
from $\PI{2}\XMATCH{3D}$ to $\PI{2}\AIKP$ by using examples.
We observe an instance 
$P_1 = (W,X,Y,N_1,N_2)$ of $\PI{2}\XMATCH{3D}$, where
\begin{eqnarray*}
&&
W=\{w_1,w_{2},w_{3},w_{4}\}, \ \ \ 
X=\{x_1,x_2,x_3,x_4\}, \ \ \ 
Y=\{y_1,y_2,y_3,y_4\},
\\
&& N_1=\{(w_{1}, x_1, y_2), (w_{2}, x_3, y_3)\},
\\ 
&& N_2=\{(w_{1}, x_1, y_1),
(w_{2}, x_2, y_1), (w_{2}, x_2, y_2),
(w_{3}, x_2, y_1), \\
&& \hspace{12mm} (w_{3}, x_2, y_2),
(w_{3}, x_3, y_3), (w_{4}, x_4, y_4)\}.
\end{eqnarray*}
From Table~\ref{tab:1}, we find that
$P_1$ is a yes instance.
For the given $P_1$, we construct
an instance $\psi(P_1) = (\Psi(P_1), \lambda(P_1))$ of $\PI{2}\AIKP$
as follows.
Table~\ref{tab:3} illustrates
the form of representations that we construct
for every triple in $N_1 \cup N_2$.
For every triple in $W \times X \times Y$,
we construct $15$-place $5$-representations.
In Subsections%
~\ref{subsubsec:simulation_matchings},%
~\ref{subsubsec:simulation_universal_quantification},%
~and~\ref{subsubsec:construction_integers_intervals},
we will describe the details of
the representations in
the $2$-nd to $4$-th columns of Table~\ref{tab:3}.

\begin{table}[htbp]
\caption{Subsets $\nu_2$ of $N_2$ for all subsets $\nu_1$ of $N_1$ such that $\nu_1 \cup \nu_2$ is a matching}
\label{tab:1}
\begin{tabular}{cc}
\hline
$\nu_1$ (subsets of $N_1$)
& $\nu_2$ (subsets of $N_2$)
\\
\hline
$\emptyset$
& $\{(w_1, x_1, y_1), (w_2, x_2, y_2),
(w_3, x_3, y_3), (w_4, x_4, y_4)\}$
\\
$\{(w_1, x_1, y_2)\}$
& $\{(w_2, x_2, y_1), (w_3, x_3, y_3), (w_4, x_4, y_4)\}$
\\
$\{(w_2, x_3, y_3)\}$
& $\{(w_1, x_1, y_1),(w_3, x_2, y_2), (w_4, x_4, y_4)\}$
\\
$\{(w_{1}, x_1, y_2), (w_{2}, x_3, y_3)\}$
& $\{(w_3, x_2, y_1), (w_4, x_4, y_4)\}$
\\
\hline
\end{tabular}
\end{table}

\begin{table}[htbp]
\caption{Constructed representations,
where the shaded rows correspond to
triples in $N_1$ and the others in $N_2$ and
$\mathsf{i,j,k} \in [0,4]$ and $\mathsf{h} \in [1,4]$}
\label{tab:3}
\begin{tabular}{cp{7em}p{7em}p{7em}}
\hline
Triples 
& $\mathrm{Addition}$ $\mathrm{limiting}$  $\mathrm{representations}$
& $\mathrm{Universal}$ $\mathrm{quantification}$ $\mathrm{testing}$ $\mathrm{representations}$
& $\mathrm{Match}$ $\mathrm{testing}$ $\mathrm{representations}$
\\
\hline
$(w_1,x_1,y_1)$
& $\overline{1}$
& $\overline{0} \ \overline{0}$ {\Large $\phantom{1}$}
& $\overline{000\mathsf{i}} \ \overline{000\mathsf{j}} \ \overline{000\mathsf{k}}$
\\ 
\rowcolor{gray} 
$(w_1,x_1,y_2)$
& $\overline{1}$ 
& $\overline{0} \ \overline{\mathsf{h}}$ {\Large $\phantom{1}$}
& $\overline{000\mathsf{i}} \ \overline{000\mathsf{j}} \ \overline{00\mathsf{k}0}$
\\
$(w_2,x_2,y_1)$
& $\overline{1}$ 
& $\overline{0} \ \overline{0}$ {\Large $\phantom{1}$}
& $\overline{00\mathsf{i}0} \ \overline{00\mathsf{j}0} \ \overline{000\mathsf{k}}$
\\
$(w_2,x_2,y_2)$
& $\overline{1}$ 
& $\overline{0} \ \overline{0}$ {\Large $\phantom{1}$}
& $\overline{00\mathsf{i}0} \ \overline{00\mathsf{j}0} \ \overline{00\mathsf{k}0}$
\\
\rowcolor{gray}  
$(w_2,x_3,y_3)$
& $\overline{1}$ 
& $\overline{\mathsf{h}} \ \overline{0}$ {\Large $\phantom{1}$}
& $\overline{00\mathsf{i}0} \ \overline{0\mathsf{j}00} \ \overline{0\mathsf{k}00}$
\\
$(w_{3},x_2,y_1)$
& $\overline{1}$ 
& $\overline{0} \ \overline{0}$ {\Large $\phantom{1}$}
& $\overline{0\mathsf{i}00} \ \overline{00\mathsf{j}0} \ \overline{000\mathsf{k}}$ 
\\
$(w_{3},x_2,y_2)$
& $\overline{1}$ 
& $\overline{0} \ \overline{0}$ {\Large $\phantom{1}$}
& $\overline{0\mathsf{i}00} \ \overline{00\mathsf{j}0} \ \overline{00\mathsf{k}0}$ 
\\
$(w_{3},x_3,y_3)$
& $\overline{1}$ 
& $\overline{0} \ \overline{0}$ {\Large $\phantom{1}$}
& $\overline{0\mathsf{i}00} \ \overline{0\mathsf{j}00} \ \overline{0\mathsf{k}00}$ 
\\
$(w_{4},x_4,y_4)$
& $\overline{1}$ 
& $\overline{0} \ \overline{0}$ {\Large $\phantom{1}$}
& $\overline{\mathsf{i}000} \ \overline{\mathsf{j}000} \ \overline{\mathsf{k}000}$ 
\\
\hline
\end{tabular}
\end{table}

\subsubsection{Simulation of Matchings}
\label{subsubsec:simulation_matchings}

In this subsection, we explain the details of the $4$th column
of Table~\ref{tab:3}.
To simulate matchings in $P_1$,
for every triple in $W \times X \times Y$,
we use $12$-place $5$-representations,
which we call {\itshape match testing representations}.
The base $5$ of a match testing representation
is equal to $q+1$, where $q$ is
the cardinality of each set of $W$, $X$, and $Y$.
This equality is for keeping the consistency in a simulation 
of a union operation in $\PI{2}\XMATCH{3D}$
by using additions of integers in $\PI{2}\AIKP$.
We describe the details below.
The $4$-th column of Table~\ref{tab:3} illustrates
the form of the match testing representations 
for every triple in $W \times X \times Y$.
A match testing representation consists of
three distinct parts: one each for $W$, $X$, and $Y$.
Each part consists of $4$ digits.
Each digit in each part corresponds to 
a variable in the corresponding set.
Let $r$ be the representation
$\overline{0\mathsf{i}00} \ \overline{00\mathsf{j}0} \
\overline{00\mathsf{k}0}$,
i.e.,
the one in 
the $4$th column and the $7$th row in Table~\ref{tab:3}.
The subrepresentation $r[9,12] = \overline{0\mathsf{i}00}$
means that $w_3$ is the third element of $W$;
$r[5,8] = \overline{00\mathsf{j}0}$
means that $x_2$ is the second of $X$;
and $r[1,4] = \overline{00\mathsf{k}0}$
means that $y_2$ is the second of $Y$.

By the above description,
we can easily observe that
if a subset $\nu$ of $W \times X \times Y$ is a matching, then
we can find a match testing representation $r(t)$ for every $t \in \nu$
such that $\sum_{t \in \nu} r(t)$ is in $\overline{[0,4]}^{12}$.
However,
the converse does not necessarily hold due to carries in additions.
For example, 
$\overline{1111} \ \overline{1111} \ \overline{1111}$
can be written as the sum of
$4$ match testing representations
\begin{center}
 $\overline{0002} \ \overline{0002} \ \overline{0002}$, \ \ \ \ 
 $\overline{0004} \ \overline{0004} \ \overline{0004}$, \ \ \ \
 $\overline{0100} \ \overline{0100} \ \overline{0100}$, \ \ \ \  
$\overline{1000} \ \overline{1000} \ \overline{1000}$
\end{center}
although the set $\{(w_1,x_1,y_1), (w_3,x_3,y_3), (w_4,x_4,y_4)\}$
of the corresponding triples is not a matching.
Next, let us consider the restriction that
the number of added representations is at most $4$.
Under this restriction,
$\overline{4444} \ \overline{4444} \ \overline{4444}$
can be written as the sum of match testing representations
only if the corresponding triples is a matching.
This property is due to the fact
that the base is greater than the number of added representations.
In addition, by construction,
if $\overline{4444} \ \overline{4444} \ \overline{4444}$
can be written as the sum of $4$ match testing representations,
then any representation in $\overline{[0,4]}^{12}$
can be also written as the sum of $4$ match testing representations.
Consequently, 
a subset of $W \times X \times Y$
is a matching if and only if
every integer in $[0,(5^{12}-1)]$
can be  written as the sum of integers
that correspond to match testing representations.

\subsubsection{Simulation of Universal Quantification}
\label{subsubsec:simulation_universal_quantification}

In this subsection, we explain the details of the $3$rd column
of Table~\ref{tab:3}.
By using the representations in this column,
we check whether, in $\PI{2}\AIKP$, every subset of $N_1$ 
forms a matching together with a subset of $N_2$.
To simulate this check, we use $2$-place $5$-representations,
which we call {\itshape universal quantification testing representations},
for every triple in $N_1$.
In a universal quantification testing representation,
the $i$-th digit corresponds to the $i$-th element of $N_1$.
In particular, if $\mathsf{k}$ is in $[1,4]$,
then $\overline{0 \mathsf{k}}$ corresponds to
$(w_1,x_1,y_2)$, which is the smallest element of $N_1$,
and $\overline{\mathsf{k} 0}$ corresponds to
$(w_2,x_3,y_3)$, which is the $2$nd smallest element of $N_1$.
We can observe the relationship
between subsets of $N_1$ and 
sums of universal quantification testing
representations in Table~\ref{tab:4}.

\begin{table}[htbp]
\caption{Sums of the universal quantification testing
representations for subsets of $N_1$.}
\label{tab:4}
\begin{tabular}{p{11em}p{11em}p{12em}}
\hline
Subsets of $N_1$
& Sums of the corresponding universal quantification testing representations
& Corresponding integers 
\\ \hline
$\emptyset$
& $\overline{00}$
& $0$ {\Large $\phantom{1}$}
\\
$\{(w_1, x_1, y_2)\}$
& $\overline{0\mathsf{k}}$ \ \ ($\mathsf{k} \in [1,4]$)
& $(\overline{0\mathsf{k}})_5 \in [1,4]$ {\Large $\phantom{1}$}
\\
$\{(w_2, x_3, y_3)\}$
& $\overline{\mathsf{k}0}$ \ \ ($\mathsf{k} \in [1,4]$)
& $(\overline{\mathsf{k}0})_5 \in \{5, 10, 15, 20\}$ {\Large $\phantom{1}$}
\\
$\{(w_{1}, x_1, y_2), (w_{2}, x_3, y_3)\}$
& $\overline{\mathsf{j}\mathsf{k}}$ \ \ ($\mathsf{j}, \mathsf{k} \in [1,4]$)
& $(\overline{\mathsf{j}\mathsf{k}})_5 \in [6,24] \backslash \{10,15,20\}$
\\
\hline
\end{tabular}
\end{table}

As in the discussion for match testing representations
in Subsection~\ref{subsubsec:simulation_matchings},
we can observe the following.
For any $\nu_1 \subseteq N_1$,
let $S(\nu_1)$ denote the set
$\{5 d_2 + d_1
\colon  d_i \in [1,4] \text{ if the }
i\text{-th element of } N_1 \text{ is in } \nu_1
\text{; and } d_i = 0
\text{ otherwise for each } i
\text{ of } 1 \text{ and } 2\}$.
Then,
$\nu_1 \subseteq W \times X \times Y$ is a subset of $N_1$
if and only if every integer in $S(\nu_1)$
can be written as the sum of integers
corresponding to the universal quantification testing representations.
Notably,
$\bigcup_{\mu \subseteq N_1} S(\mu)$
is an interval
although $S(\nu_1)$ may not be an interval
for any $\nu_1 \subseteq N_1$.
We can observe that the set of integers
at all the cells in the $3$rd column of Table~\ref{tab:4}
is the interval $[0,24]$.

Moreover, 
we define the universal quantification testing representations for every
triple of $N_2$ as $\overline{00}$.

\subsubsection{Construction of Integers and Intervals}
\label{subsubsec:construction_integers_intervals}

In this subsection, we explain the details of the $2$nd column of Table~\ref{tab:3}.
For simulating an instance of $\PI{2}\XMATCH{3D}$,
we construct $15$-place $5$-representations such that
the following condition is satisfied.
If $R$ is a set of $15$-place $5$-representations that corresponds to a matching,
then the sum of all elements of $R$ is in 
\begin{center}
$[\overline{4} \ \overline{00} \ \overline{0000}
\ \overline{0000} \ \overline{0000} \ , \
\overline{4} \ \overline{44} \ \overline{4444}
\ \overline{4444} \ \overline{4444}]$.
\end{center}
For this purpose,
for every $15$-place $5$-representation,
we use a single digit $\overline{1}$, which we call
its {\itshape addition limiting representation}.
We place it as the most significant digit
of every $15$-place $5$-representation.
Then,
we define $\Psi(P_1)$ as the set of all $15$-place $5$-representations 
$\overline{1} r_\mathrm{u}(t) r_\mathrm{m}(t)$,
where
$t \in W \times X \times Y$
and $r_\mathrm{u}(t)$ and $r_\mathrm{m}(t)$
are universal quantification and match testing representations for
$t$, respectively.
In addition, we define $\lambda(P_1)$ as
$(\overline{4} \ \overline{00} \ \overline{0000}
\ \overline{0000} \ \overline{0000})_5$,
i.e., $4 \cdot 5^{14}$.
Moreover, the smallest element of $\Psi(P_1)$ is
$(\overline{1} \ \overline{00} \ \overline{0000}
\ \overline{0000} \ \overline{0000})_5$,
i.e., $5^{14}$.

Given a set $\tau$ of triples in $W \times X \times Y$,
for verifying whether $\tau$ is a matching,
it suffices to check whether the sum of all elements in $\Psi(\tau)$
is in $[4 \cdot 5^{14},(5^{15}-1)]$.
Table~\ref{tab:5} illustrates a set of integers constructed
from a matching in $P_1$.
We can observe the following.
For a subset $\{(w_2,x_3,y_3)\}$ of $N_1$,
there is a subset 
\begin{center}
$\{(w_1,x_1,y_1), \ \ (w_3,x_2,y_2), \ \ (w_4,x_4,y_4)\}$ 
\end{center}
of $N_2$ such that
the union is a matching.
The sum of the $4$ integers in the $1$st column is
$(\overline{4} \ \overline{40} \ \overline{4444} \ \overline{4444} \
\overline{4444})_5$.
This sum is in $[4 \cdot 5^{14},(5^{15}-1)]$.

\begin{table}[htbp]
\caption{Example of a tuple of constructed integers whose sum
corresponds to a matching} 
\label{tab:5}
\begin{tabular}{cc}
\hline
Constructed integers & Corresponding triples
\\
\rowcolor{gray}  
$(\overline{1}$ \ $\overline{40}$ \ $\overline{0040}$ \ $\overline{0040}$ 
\ $\overline{0400})_5$ {\Large $\phantom{1}$}
& $(w_2,x_3,y_3) \in N_1$ 
\\
$(\overline{1}$ \ $\overline{00}$ \ $\overline{0004}$ \ $\overline{0004}$ 
\ $\overline{0004})_5$ {\Large $\phantom{1}$}
& $(w_1,x_1,y_1) \in N_2$ 
\\
$(\overline{1}$ \ $\overline{00}$ \ $\overline{0400}$ \ $\overline{0040}$ 
\ $\overline{0040})_5$ {\Large $\phantom{1}$}
& $(w_3,x_2,y_2) \in N_2$
\\
$(\overline{1}$ \ $\overline{00}$ \ $\overline{4000}$ \ $\overline{4000}$ 
\ $\overline{4000})_5$ {\Large $\phantom{1}$}
& $(w_4,x_4,y_4) \in N_2$
\\
\hline
\end{tabular}
\end{table}

\subsection{The $\PIP{2}$-Completeness of $\PI{2}\AIKP$}
\label{subsec:formulation_reduction_3DM}

In this subsection,
we prove $\PI{2}\AIKP$ to be $\PIP{2}$-complete under Karp reductions.
As the hardness part of the proof,
we reduce $\PI{2}\XMATCH{3D}$ to $\PI{2}\AIKP$
as informally described in Subsection~\ref{subsec:ideas_reduction_Pi3DM}.
$\PI{2}\XMATCH{3D}$ is known to be $\PIP{2}$-complete under Karp reductions~\cite{1056978}.

\begin{theorem}
\label{thm:1}
$\PI{2}\AIKP$ is $\PIP{2}$-complete under Karp reductions.
\end{theorem}

\begin{proof}
We first prove $\PI{2}\AIKP$ to be in $\PIP{2}$.
More generally,
we show that $\PI{2}\IKP$ is in $\PIP{2}$.
Let $(A,\lambda,\upsilon)$ be an instance of $\PI{2}\IKP$.
By definition, every integer in the interval $[\lambda,\upsilon]$
can be represented as a binary representation of polynomial length in the input.
Moreover, for every integer $k$ in $[\lambda,\upsilon]$,
we can check whether there are nonnegative integers $x_1,\cdots,x_n$ such that
$\sum_{i=1}^n x_i a_i = k$, where $\{a_1,\cdots,a_n\} = A$,
in polynomial time
by using $\IKP$~(Section 15.7 in \cite{Papadimitriou:1982:COA:31027}) as an oracle.
$\IKP$ is known to be an $\NP$-complete problem%
~(Section 15.7 in \cite{Papadimitriou:1982:COA:31027}).
Thus, $\PI{2}\IKP$ is in $\PIP{2}$.
 
Then, 
we will prove $\PI{2} \AIKP$ to be $\PIP{2}$-hard under Karp reductions.
For this purpose, 
we reduce $\PI{2}\XMATCH{3D}$ to $\PI{2}\AIKP$.
In particular, 
we formulate the mapping $\psi$ as described in Subsection~\ref{subsec:ideas_reduction_Pi3DM}.
Let $P = (U_1,U_2,U_3,M_1,M_2)$
be fixed to an instance of $\PI{2}\XMATCH{3D}$.
Let $q$ be the cardinality of $U_1$,
i.e., $|U_1|=|U_2|=|U_3|=q$.
Then, we will define an instance $\psi(P) = (\Psi(P), \lambda(P))$
of $\PI{2}\AIKP$,
where $\Psi(P)$ and $\lambda(P)$ are defined as follows.

For every $1 \leq i \leq 3$,
let $u_{i,1}, \cdots, u_{i,q}$ be all elements of $U_i$.
Let $b$ be integer $q+1$.
The integer $b$ is the basis of all representations that we construct.
Let $t = (u_{1,j_1}, u_{2,j_2}, u_{3,j_3})$
be a triple in $U_1 \times U_2 \times U_3$,
where $j_1,j_2,j_3 \in [1,q]$.
Then, we define $R_\mathrm{m}(t)$ as the set of all $b$-representations 
of the form 
\begin{center}
 $\underbrace{\overline{0} \ \ \cdots \ \ \overline{0}}_{(q-j_1) \text{ digits}} \
 \overline{d}_1 \
 \underbrace{\overline{0} \ \ \cdots \ \  \overline{0}}_{(j_1-1) \text{ digits}} \ 
 \underbrace{\overline{0} \ \ \cdots \ \ \overline{0}}_{(q-j_2) \text{ digits}} \
 \overline{d}_2 \
 \underbrace{\overline{0} \ \ \cdots \ \ \overline{0}}_{(j_2-1) \text{ digits}} \
 \underbrace{\overline{0} \ \ \cdots \ \ \overline{0}}_{(q-j_3) \text{ digits}} \
 \overline{d}_3 \
 \underbrace{\overline{0} \ \ \cdots \ \ \overline{0}}_{(j_3-1) \text{ digits}}$, 
\end{center}
where $d_1,d_2,d_3 \in [0,q]$.
Let $t_1, \cdots, t_{|M_1|}$ be all elements of $M_1$,
where $t_1 < \cdots < t_{|M_1|}$.
For every $1 \leq k \leq |M_1|$,
we define $R_\mathrm{u}(t_k)$ as the set of
all $b$-representations of the form
\begin{center}
$\underbrace{\overline{0} \ \ \cdots \ \ \overline{0}}_{(|M_1|-k) \text{ digits}} \
\overline{d} \
\underbrace{\overline{0} \ \ \cdots \ \ \overline{0}}_{(k-1) \text{ digits}}$,
\end{center}
where $d \in [1,q]$.

For every $t \in (U_1 \times U_2 \times U_3) \backslash M_1$,
we define $R_\mathrm{u}(t)$ as the set $\{\overline{0}^{|M_1|}\}$.
For any $s \in U_1 \times U_2 \times U_3$,
we define $\Gamma(s)$ as the set
$\{\overline{1} \ \alpha \ \beta \colon
\alpha \in R_\mathrm{u}(s), \beta \in R_\mathrm{m}(s)\}$.
Let $r$ be a representation in
$\bigcup_{s \in U_1 \times U_2 \times U_3} \Gamma(s)$.
Then,
we call $r[1,3q]$, $r[3q+1,3q+|M_1|]$, and $r[3q+|M_1|+1]$
the {\itshape match testing}, {\itshape universal quantification testing}, and
{\itshape addition limiting} representations
for $r$, respectively.
For any $s \in U_1 \times U_2 \times U_3$,
we define 
$\Psi(s)$ as the set $\{(\alpha)_b  \colon \alpha \in \Gamma(s)\}$.
For any $S \subseteq U_1 \times U_2 \times U_3$,
we define 
$\Psi(S)$ as the set $\bigcup_{s \in S} \Psi(s)$.
We define $\Psi(P)$ as the set $\Psi(M_1) \cup \Psi(M_2)$.
We define $\lambda(P)$ as 
integer $(\overline{q} \ \overline{0}^{3q+|M_1|})_b$.

\begin{algorithm}
\SetAlgoNoLine
\KwIn{$P = (U_1,U_2,U_3,M_1,M_2)$.}
\KwOut{$\psi(P)$.}
\For {each set $S$ of $M_1$, $M_2$, $U_1$, $U_2$, and $U_3$}{
Sort all the elements of $S$ in ascending order}
$A \longleftarrow \emptyset$\;
\For{every triple $t=(u_1,u_2,u_3) \in M_1 \cup M_2$}{
\For{every $i_1,i_2,i_3 \in [1,q]$}{
\For{every $d_1$, $d_2, d_3 \in [0, q]$}{ 
\uIf{$t \in M_1$}{ 
$i_0 \longleftarrow $ (an integer such that $t$ is the $i_0$-th smallest element of $M_1$)\;
\For{every $d_0 \in [1, q]$}{
 $c \longleftarrow 
 (\overline{1} \ 
 \overline{0}^{|M_1|-i_0} \ \overline{d_0} \ 
 \overline{0}^{i_0 - i_1 + q - 1} \ \overline{d_1} \
 \overline{0}^{i_1 - i_2 + q - 1} \ \overline{d_2} \
 \overline{0}^{i_2 - i_3 + q - 1} \ \overline{d_3} \
\overline{0}^{i_3-1})_b$\;
$A \longleftarrow A \cup \{c\}$\;
}
}
\Else{
$c \longleftarrow 
(\overline{1} \ 
\overline{0}^{|M_1| - i_1 + q - 1} \ \overline{d_1} \
\overline{0}^{i_1 - i_2 + q - 1} \ \overline{d_2} \
\overline{0}^{i_2 - i_3 + q - 1} \ \overline{d_3} \
\overline{0}^{i_3-1})_b$\;
$A \longleftarrow A \cup \{c\}$\;
}
}
}
}
$e \longleftarrow \min A$\;
\KwRet $(A,e)$\;
\caption{Reduction from $\PI{2}\AIKP$ to $\PI{2}\XMATCH{3D}$}
\label{alg:1}
\end{algorithm}

Given $P$,
we can compute $\psi(P)$ in polynomial time as Algorithm~\ref{alg:1}.
In Algorithm~\ref{alg:1},
the innermost loop is in lines $10$-$13$.
By definition, $|M_1 \cup M_2|$ is less than or equal to $q^3$.
Thus, the loop at lines $5$-$20$ is repeated at most $q^3$ times.
In each iteration of the loop at lines $5$-$20$,
the loop at lines $6$-$19$ is repeated at most $(q-1)^3$ times.
In each iteration of the loop at lines $6$-$19$,
the loop at lines $7$-$18$ is repeated at most $q^3$ times.
In each iteration of the loop at lines $7$-$18$,
the loop at lines $10$-$13$ is repeated at most $(q-1)^3$ times.
Thus, we can compute $\psi(P)$ in time with polynomial order in $q$.

In the remainder of the proof,
we confirm the validity of the reduction.
In particular,
we prove that
$P$ is a yes instance of $\PI{2}\XMATCH{3D}$ if and only if
$\psi(P)$ is a yes instance of $\PI{2}\AIKP$.
We first show the {\itshape ``only if''} part  of the proof.
Let $\mu_1 \subseteq M_1$
and $\mu_2 \subseteq M_2$
such that $\mu_1 \cup \mu_2$ is a matching.
Let $t_1,\cdots,t_\xi$ be
all elements of $\mu_1$,
where $t_1 < \cdots < t_\xi$.
Let $t_{\xi+1},\cdots,t_q$ be
all elements of $\mu_2$,
where $t_{\xi+1} < \cdots < t_q$.
Let $r$ $=$ $\overline{d_{\alpha+1}}$ $\cdots$ $\overline{d_1}$
be a $b$-representation such that 
$\overline{q} \overline{0}^\alpha \leq r \leq \overline{q}^{\alpha+1}$,
where $\alpha = 3q+|M_1|$.
Then, it suffices to show the following statement (I).

\noindent 
(I) {\itshape  
We can find
$r_k \in \Gamma(t_k)$ for every $1 \leq k \leq q$
such that $\sum_{k=1}^q r_k = r$.}

By construction,
for every $1 \leq k \leq \xi$, we can find $r_k \in \Gamma(t_k)$ of the form
\begin{eqnarray*}
&& \overline{1} \ 
\underbrace{\overline{0} \ \ \ \ \ \cdots \ \  \ \ \ \overline{0}}_{(|M_1|-k_0) \text{ digits}} \
\overline{d_{3q+k_0}} \ 
\underbrace{\overline{0} \ \ \ \ \ \cdots \ \  \ \ \ \overline{0}}_{(q+k_0-k_1-1) \text{ digits}} \
\overline{d_{2q+k_1}} \ 
\\
&& \ \ \ \ \ \ \ \ \ \ \ \ \ \ \ \ \ \ \ 
\underbrace{\overline{0} \ \ \ \ \ \cdots \ \  \ \ \ \overline{0}}_{(q+k_1-k_2-1)  \text{ digits}} \
\overline{d_{q+k_2}} \ 
\underbrace{\overline{0} \ \ \ \ \ \cdots \ \  \ \ \ \overline{0}}_{(q+k_2-k_3-1)  \text{ digits}} \
\overline{d_{k_3}} \ 
\underbrace{\overline{0} \ \ \ \ \ \cdots \ \  \ \ \ \overline{0}}_{(k_3-1)  \text{ digits}},
\end{eqnarray*}
where $t_k$ is the $k_0$-th smallest element of $M_1$,
and for every $1 \leq i \leq 3$,
the $i$-th component of $t_k$ is the $k_i$-th smallest element of $U_i$.
Similarly, by construction,
for every $\xi+1 \leq k \leq q$, we can find $r_k \in \Gamma(t_k)$ of the form
\begin{eqnarray*}
&& \overline{1} \ 
\underbrace{\overline{0} \ \ \ \ \ \cdots \ \  \ \ \ \overline{0}}_{(q+|M_1|-k_1) \text{ digits}} \
\overline{d_{2q+k_1}} \ 
\underbrace{\overline{0} \ \ \ \ \ \cdots \ \  \ \ \ \overline{0}}_{(q+k_1-k_2-1)  \text{ digits}} \
\\
&& \ \ \ \ \ \ \ \ \ \ \ \ \ \ \ \ \ \ \ \ \ \ \ \ \ \ \ \ \ \ \ \ \ \ 
\overline{d_{q+k_2}} \ 
\underbrace{\overline{0} \ \ \ \ \ \cdots \ \  \ \ \ \overline{0}}_{(q+k_2-k_3-1)  \text{ digits}} \
\overline{d_{k_3}} \ 
\underbrace{\overline{0} \ \ \ \ \ \cdots \ \  \ \ \ \overline{0}}_{(k_3-1)  \text{ digits}},
\end{eqnarray*}
where for every $1 \leq i \leq 3$,
the $i$-th component of $t_k$ is the $k_i$-th smallest element of $U_i$.
By the assumption that $\mu_1 \cup \mu_2$ is a matching,
the following holds.
For every $i$ and $j$ with $1 \leq i < j \leq q$,
if $c_i$ and $c_j$ are the $l$-th components of 
$t_i$ and $t_j$ for some $1 \leq l \leq 3$, respectively,
then $c_i \neq c_j$
Thus, for every $i$ and $j$ with $1 \leq i < j \leq q$,
there is no $1 \leq p \leq 3q+|M_1|$ such that
$r_i[p] \neq \overline{0}$ and $r_j[p] \neq \overline{0}$.
Consequently, statement (I) holds.

Next, we show the {\itshape ``if''} part.
The proof for this part requires more careful arguments.
We prove the following statement.
Let $n_1,\cdots,n_\kappa$ denote 
integers in $\Psi(P)$,
where $n_1 \leq \cdots \leq n_\kappa$
and $\kappa \in \mathbb{N}_+$.
Let $I$ be a subset of $[3q+1,3q+|M_1|]$.

(II)
{\itshape If
$\sum_{k=1}^\kappa n_k =
\sum_{k \in [1,3q] \cup I \cup \{3q + |M_1|\}} q b^{k-1}$,
then we can find a matching
$\{t_1, \cdots, t_\kappa\}$
such that $n_i \in \Psi(t_i)$
for every $1 \leq i \leq \kappa$.}

For every $i \in [1,\kappa]$,
let $r_i$ denote the $b$-representation of $n_i$.
Let $r$ denote the $\sum_{i=1}^\kappa r_i$.
To prove statement (II), 
we show the following statements.

(III) {\itshape$\kappa = q$.}

(IV) {\itshape For every $p \in [1,3q] \cup I$,
there is exactly one $r_k$,
where $k \in [1,\kappa]$, such that
$r_k[p] \neq \overline{0}$.}

By definition,
if statement (IV) is satisfied, then
we can obtain a matching
$\{t_1, \cdots, t_\kappa\}$
from $r_1, \cdots, r_\kappa$.
Let us prove statement (III).
By definition,
the most significant digit of any representation in $\Gamma$
is $\overline{1}$.
Thus, since $r[3q+|M_1|+1] = \overline{q}$, 
$\kappa$ is less than or equal to $q$.
By definition, for every $1 \leq i \leq \kappa$,
the match testing representation of $r_i$
consists of at most $3$ non-zero digits.
Moreover, by definition,
the match testing representation of $r$
is $\overline{q}^{3q}$.
Thus, $\kappa$ is greater than or equal to $q$.
Consequently, $\kappa = q$.
Next, assume that
there is $j \in [1,3q] \cup I$ such that
$r_k[j] = \overline{0}$ for every $k \in [1,\kappa]$.
Then, at most $\kappa-1$ carries occur 
at the $(j-1)$-th digits
in additions of $n_1,\cdots,n_\kappa$.
It follows
that $r[j]$ is at most $\kappa-1$,
i.e., at most $q-1$ by statement (III).
This contradicts the assumption that
every digit of $r$ is $\overline{q}$.
Thus, statement (IV) holds.
The proof of Theorem~\ref{thm:1} is complete.
\end{proof}

\section{The $\SIGMAP{2}$-Completeness of $\FROB$ under Karp Reductions} 
\label{sec:completeness_FROB}

In this section, we prove the $\SIGMAP{2}$-completeness of $\FROB$ under Karp reductions.
This is the main theorem of this paper.
We will obtain this theorem as a corollary of a theorem that
$\COFROB$ is $\PIP{2}$-complete under Karp reductions.
In Subsection~\ref{subsec:difference_IKP_FP},
we observe differences in instances of $\PI{2}\AIKP$ and $\COFROB$.
In Subsection~\ref{subsec:ideas_reduction_FP_IKP},
we describe the key concepts of our reduction from $\PI{2}\AIKP$ and $\COFROB$.
In Subsection~\ref{subsec:formulation_reduction_FP_IKP},
we prove the theorems.

\subsection{Differences in instances of $\PI{2}\AIKP$ and $\COFROB$}
\label{subsec:difference_IKP_FP}

The definition of $\PI{2}\AIKP$ is similar to that of $\COFROB$.
However, there are three main differences.
The first difference is whether given positive integers may include integer $1$.
An instance of $\PI{2}\AIKP$ may include $1$,
although all given positive integers are greater than or equal to $2$ in $\COFROB$.
The second difference is the number of given positive integers.
The number is greater than or equal to $0$ in $\PI{2}\AIKP$,
whereas the number is greater than or equal to $2$ in $\COFROB$.
The third difference is whether all given positive integers are coprime.
They may be not coprime in $\PI{2}\AIKP$,
whereas they are subject to be coprime in $\COFROB$.

\subsection{Main Ideas of Our Reduction}
\label{subsec:ideas_reduction_FP_IKP}

In this subsection, we describe the main ideas of our reduction.
Let $Q = (A, \lambda)$ be an instance of $\PI{2}\AIKP$,
where $|A| = n$ for some $n \geq 0$.
Let $a_1,\cdots,a_n$ be all elements of $A$,
where $1 \leq a_1 < \cdots < a_n$.
If $a_1 \geq 2$, $|A| \geq 2$, and 
$a_1,\cdots,a_n$ are coprime, 
then $Q$ is also an instance of $\COFROB$.

If $a_1 = 1$, then $Q$ is a yes instance,
not depending on the other elements of $A$ and the integer $\lambda$.
In this case,
any positive integer can be written as
a multiple of $a_1$, i.e.,
as a nonnegative integer combination of $A$.
Thus, in this case, it suffices to correspond $Q$ to
a yes instance of $\COFROB$.

In the case where $a_1 \neq 1$,
if $|A| = 1$ or $a_1,\cdots,a_n$ are not coprime,
then $Q$ is a no instance of $\PI{2}\AIKP$.
For example, if $A = \{2\}$, i.e., $a_1 = 2$,
then no odd integer can be written as
a multiple of $a_1$.
As another example,
we observe the case where $A = \{3,6,9\}$ and $\lambda = 11$.
Then, the interval that we should check is $[11,12,13]$.
However, we cannot represent integers $11$ and $13$
as a nonnegative integer combination of $A$
since these are not multiples of $k$.
Hence, $Q$ is a no instance.
Thus, in the case where $a_1 \neq 1$,
if $|A| = 1$ or $a_1,\cdots,a_n$ are not coprime,
then it suffices to correspond $Q$ to
a no instance of $\COFROB$.

In $\COFROB$,
if the number of input integers
is two, then
we know the Frobenius number by a formula~(Section~2.1 in \cite{ramirezAlfonsin2005diophantine}).
Furthermore, given any instance $(A,k)$,
every integer greater than $g(A)$
can be represented as a nonnegative integer combination of $A$.
Thus, we can easily find both yes instances and no instances.
For example,
we observe the case in which $\{3,4\}$ is a given input.
Then, since the Frobenius number is $3\cdot4 - 3 -4 = 5$,
the integer $5$ 
cannot be represented as a nonnegative integer combination $\{3,4\}$,
and each of the integers $6,7,8$ 
can be represented as a nonnegative integer combination of  $\{3,4\}$.
Thus, we find $(\{3,4\}, 6)$ to be a yes instance and
$(\{3,4\}, 5)$ to be a no instance.

In $\COFROB$, an instance $(B,k)$ is a yes instance
if all integers in a ``sufficiently'' large interval
can be represented as a nonnegative integer combination of $B$.
We do not know the smallest sufficient length of intervals.
However, if all integers in $[l,l+\min B-1]$
can be represented as a nonnegative integer combination of $B$
for some integer $l$,
then $g(B)$ is less than $l$.

\subsection{A Main Theorem}
\label{subsec:formulation_reduction_FP_IKP}

In this subsection, by showing the following theorem,
we obtain the $\SIGMAP{2}$-completeness of $\FROB$ under Karp reductions.

\begin{theorem}
\label{thm:5}
$\COFROB$ is $\PIP{2}$-complete under Karp reductions. 
\end{theorem}
\begin{proof}
 We first prove the membership of $\COFROB$ to $\PIP{2}$.
 Let $(A,k)$ be an instance of $\COFROB$.
 We denote all elements of $A$ by $a_1, \cdots, a_n$.
 Let $U$ be $a_n^2$. 
 $U$ is known to be an upper bound of $g(A)$ due to Wilf~\cite{Wilf1978}.
 In other words, to check whether every integer $m \geq k$  
 can be represented as a nonnegative integer combination of $A$, 
 it suffices to check all integers that are less than or equal to $U$.
 By definition, $U$ can be represented
 as a binary representation whose length is of polynomial order in the size of $(A,k)$.
 Let $(A,k)$ be a yes instance. 
 Then, for all integers $m$ in the interval $[k,U]$,
 we can determine whether $m$  
 can be represented as a nonnegative integer combination of $A$ 
 in polynomial time 
 by using $\IKP$ as its oracle.
 $\IKP$ is known to be an $\NP$-complete problem%
 ~(Section 15.7 in \cite{Papadimitriou:1982:COA:31027}).
 Consequently, $\COFROB$ is in $\PIP{2}$.

 Then, we prove the $\PIP{2}$-hardness of $\COFROB$
 by constructing a reduction from $\PI{2}\AIKP$.
 In particular,
 we formulate the reduction
 informally described in Subsection~\ref{subsec:ideas_reduction_FP_IKP}.
 Let $Q = (A, \lambda)$ be an instance of $\PI{2}\AIKP$.
 Let $a_1,\cdots,a_n$ be all elements of $A$, 
 where $1 \leq a_1 < \cdots < a_n$.
 We construct an instance $\phi(Q) = (\phi(A), \phi(\lambda))$  
 of $\COFROB$, where $\phi(Q)$ is defined as follows.
 \begin{eqnarray*}
 (\phi(A), \phi(\lambda)) =
 \begin{cases}
 (A,\lambda) & \text{if }a_1 \geq 2; 
  |A| \geq 2; 
 \text{and } \text{all elements of } A \text{ are coprime}, 
 \\
 (\{3,4\},6) &
 \text{if } a_1 = 1,
 \\
 (\{3,4\},5) &
 \text{otherwise}.
 \end{cases}
 \end{eqnarray*}

 We can determine whether $a_1 \geq 2$ and
 whether $|A| \geq 2$ in linear time.
 Furthermore, we can check the coprimality of $A$ 
 in polynomial time since we can compute its coprimality
 by using Euclid's algorithm at most $n$ times.
 Euclid's algorithm can be executed in 
 polynomial time~(Section 4.5.2 in \cite{KnuthArtComputerProgrammingSeminumericalAlgorithmsEd3}).
 Consequently, given $Q$, we can construct $\varphi(Q)$ in polynomial time.
 By the description in Subsection~\ref{subsec:ideas_reduction_FP_IKP},
 the reduction is valid.
 The proof of Theorem~\ref{thm:5} is complete.
\end{proof}

As a corollary, we obtain the main theorem.

\begin{theorem}
\label{thm:4} 
$\FROB$ is $\SIGMAP{2}$-complete under Karp reductions.
\end{theorem}

\section{The Complexity of the Original Version of the Frobenius Problem}
\label{sec:completeness_FFROB}

In this section, we describe the computational complexity
of the original version of the Frobenius problem $\FFROB$.
Theorem~\ref{thm:4} immediately 
implies an upper bound for the complexity of $\FFROB$,
which is the first nontrivial upper bound. 
Moreover, we can also obtain an improved lower bound.
Additionally, we discuss the further improvement of these bounds.

\subsection{Upper Bounds for the Complexity Class of $\FFROB$}
\label{subsec:upper_bound_EFROB_and_FROBF}

By using Theorem~\ref{thm:4},
we can derive an upper bound for the complexity class of $\FFROB$
using usual methods in complexity theory as follows.

\begin{theorem}
\label{thm:3}
$\FFROB$ is in $\FUNC\DELTAP{3}$.
\end{theorem}
\begin{proof}
Recall that $\FFROB$ is identical to $g$.
For proving the theorem,
it suffices to show that
we construct a polynomial-time algorithm
that outputs $g(A)$ with a $\SIGMAP{2}$-oracle $L$
for a given input $A$.
Let $a_1,\cdots,a_n$ be all elements in $A$,
where $a_1 < \cdots < a_n$.
Let $U_A$ be $a_n^2$, 
which is an upper bound for $g(A)$~\cite{Wilf1978}.
We find the Frobenius number $g(A)$ among integers in $[a_1,U_A]$
by the binary search in Algorithm~\ref{alg:2}.

\begin{algorithm}
\SetAlgoNoLine
\KwIn{$A = \{a_1,\cdots,a_n\}$.}
\KwOut{$g(A)$.}
$l \longleftarrow a_1$\;
$u \longleftarrow U_A$\;
$v \longleftarrow \lfloor (l+u)/2 \rfloor$\;
\While{$u \neq l$}{
\uIf{$g(A) \geq v$}{$l \longleftarrow v+1$\;}
\Else{$u \longleftarrow v-1$\;}
$v \longleftarrow \lfloor (l+u)/2 \rfloor$\;
}
\KwRet $v$\;
\caption{Binary Search for the Frobenius Number}
\label{alg:2}
\end{algorithm}

For the comparison $g(A) \geq v$ at Step $3$,
we use $\FROB{}$ as an oracle.
By this oracle, we can run the above algorithm 
in polynomial time.
By Theorem~\ref{thm:4},
$\FFROB$ is in $\FP^{\SIGMAP{2}}$, i.e., $\FUNC\DELTAP{3}$.
\end{proof}

\subsection{Lower Bounds for the Complexity Class of $\FFROB$}
\label{subsec:lower_bound_EFROB_and_FROBF}

By Theorem~\ref{thm:4},
$\FFROB$ is at least as hard as any $\SIGMAP{2}$ problem.
This is because 
we can immediately determine whether $g(A) \geq k$
by computing $g(A)$, i.e., solving $\FFROB$ for a given instance $(A,k)$ of $\FROB$.
However, there should be a better lower bound for the complexity of $\FFROB$.
Ideally, we are expected to find a complexity class of functions as a better lower bound.

We obtained $\FUNC\DELTAP{3}$ as an upper bound
in Subsection~\ref{subsec:upper_bound_EFROB_and_FROBF}
Thus, it is natural to ask whether $\FFROB$ is $\FUNC\DELTAP{3}$-hard.
This question appears to be not easy.
Known $\FUNC\DELTAP{3}$-complete problems are few~\cite{doi:10.1137120904524}.
Moreover, the structures of such $\FUNC\DELTAP{3}$-complete problems and $\FFROB$
are quite different. 
Thus, constructing a reduction should require some sophisticated techniques.
A class of functions, $\SIGMAMM{2}$, is another candidate of a lower bound,
which was introduced in \cite{KRENTEL1992183}.
This class is a subclass of $\FUNC\DELTAP{3}$.
Every function of $\SIGMAMM{2}$ is specified by
a type of alternating Turing machine, called a {\itshape polynomial $k$-alternating max-min Turing machine}. 
Unfortunately, known results are scarce.

Natural characterizations of $\FUNC\DELTAP{k}$ and $\SIGMAP{k}$ 
are given as generalizations for other complexity classes~\cite{KRENTEL1992183}.
However, some classes of functions have only characterizations by
specific structures. 
For example,
$\PLS$~\cite{JOHNSON198879} and $\PPAD$~\cite{PAPADIMITRIOU1994498} 
are known to be such classes of functions.
$\PLS$ and $\PPAD$
are known to have complete problems~\cite{JOHNSON198879,PAPADIMITRIOU1994498}.
By an analogy of $\PLS$ and $\PPAD$,
we may obtain a suitable complexity class for $\FFROB$.
However, this paper does not cover this approach since
it appears to require a brand-new investigation for the property of $\FFROB$.

By the above discussion,
finding a function class that is harder than $\SIGMAP{2}$ is not easy.
The next best approach is 
finding a decision problems class that is harder than $\SIGMAP{2}$ and easier than $\FUNC\DELTAP{3}$.
Thus, it is worthwhile to investigate a problem $\EFROB$ and a class $\DIFFP{2}$ of decision problems. 
By definition, the class $\DIFFP{2}$ is harder than $\SIGMAP{2}$.
$\EFROB$ is expected to be harder than $\FROB$ by an analogy 
to relationships between three variants of the traveling salesman problem%
~\cite{PAPADIMITRIOU1984244}.
The first variant is deciding the existence of a Hamilton path
whose cost is at most $k$ for a given graph $G$ and a cost $k$.
The second one is deciding the existence of a Hamilton path
whose cost is exactly $k$ for a given graph $G$ and a cost $k$.
The third one is computing a Hamilton path
whose cost is minimum for a given graph $G$.
The third one is known to be harder than the second,
and the second is known to be harder than the first~\cite{PAPADIMITRIOU1984244}.
Moreover, the second one was proven to be $\DP$-complete~\cite{PAPADIMITRIOU1984244}.
If we prove the $\DIFFP{2}$-hardness of $\EFROB$,
then we may obtain a better lower bound for $\FFROB$.
However, the analysis for the $\DIFFP{2}$-hardness of $\EFROB$ is not easy.
In this paper,
we prove only the membership of $\EFROB$ to $\DIFFP{2}$. 

\begin{theorem}
\label{thm:2}
$\EFROB$ is in $\DIFFP{2}$.
\end{theorem}
\begin{proof}
It suffices to show that there are two problems 
$L_1 \in \SIGMAP{2}$ and $L_2 \in \PIP{2}$
such that $L_1 \cap L_2 = \EFROB$.
Let $L_1$ be $\FROB{}$.
Let $L_2$ be a language such that
$(A, k) \in L_2$ if and only if
$(A, k+1) \in \COFROB$.
Obviously, the language $L_2$ is in $\PIP{2}$.
Let $(A,k)$ be an instance of $\EFROB$.
Then, $(A,k)$ is a yes instance of $\EFROB$, i.e., $g(A) = k$
if and only if $(A,k)$ and $(A,k+1)$ are yes instances of $L_1$ and $L_2$,
i.e., $g(A) \geq k$ and $g(A) < k+1$.
It follows that $\EFROB$ is in $\DIFFP{2}$.
\end{proof}

We describe some reasons for the difficulties of 
the proof or disproof for the $\DIFFP{2}$-hardness
of $\EFROB$.
The existing $\DIFFP{2}$-complete problems
can be categorized into two types.
The first is a type of problem such as
$\SIGMAPISAT{2}$~\cite{WOOLDRIDGE200427}.
The second is a type of problem such as
$\INCOMPLETEGAME$~\cite{WOOLDRIDGE200427}.
Any instance of $\SIGMAPISAT{2}$~\cite{WOOLDRIDGE200427}
is specified by a pair of instances of $\SIGMAP{2}$- and $\PIP{2}$-problems.
This type of $\DIFFP{2}$-complete problem is immediately obtained
from a number of $\SIGMAP{2}$- or $\PIP{2}$-complete problems~\cite{schaefer2002completeness}.
We can prove that $\FROB\text{-}\COFROB$ is
$\DIFFP{2}$-complete by a reduction 
similar to that in the proof of Theorem~\ref{thm:4},
Any reduction from this type of problem to $\EFROB$ 
appears to require an outstanding result
in number theory.

The second type of problem is problems such as a cooperative game in~\cite{WOOLDRIDGE200427},
called $\INCOMPLETEGAME$. Three variants of package recommendation problems in~\cite{doi:10.1137120904524}
are also of this type.
This type of problem has the following properties.
Any instance of such a problem is specified by a complicated formulation.
For example, an instance of 
$\INCOMPLETEGAME$ has two types of variables, called agents and goals,
and a relation between the two sets of variables.
Such formulations are natural and reasonable
since those studies were motivated by concerns for the computational complexities
of problems in a specific research domain.
Indeed,
\cite{WOOLDRIDGE200427} studied the computational complexities of $10$ or more cooperative games 
including $\INCOMPLETEGAME$.
\cite{doi:10.1137120904524} investigated many properties of recommendation systems
from a complexity theoretical perspective.
In compensation for a sufficiently powerful description,
the formulation of this type of problem is too complicated 
to use as a reduced problem in a proof for the hardness of another problem.
That is, the reduction from this type of problem to $\EFROB$ may require
a novel technique of formulation or patience for a long description.

\section{The Strength of the $\SIGMAP{2}$-Hardness of the Frobenius Problem}
\label{sec:strength}

In this section,
we analyze the $\SIGMAP{2}$-hardnesses of $\FROB$ and $\FFROB$ more precisely.
In Section~\ref{sec:completeness_FROB},
we proved the $\SIGMAP{2}$-completeness of $\FROB$.
However,
many researchers have considered solving $\FFROB$ practically fast,
as described in Subsection~\ref{subsec:1_2}.
Their algorithms are exact algorithms.
Thus, the approximability of $\FFROB$ is an interesting subject.
We do not know the existence
of an approximation algorithm or any inapproximability for $\FFROB$.
However, there is a pseudopolynomial algorithm for $\FFROB$.
To the best of the author's knowledge,
an explicit statement has not been provided,
but some reports in the literature implicitly suggest the existence of
a pseudopolynomial algorithm, 
e.g., \cite{nijenhuis1979,beihoffer2005faster,BockerSebastiAndLiptakZsuzsanna2007}.
By this fact and Theorem~\ref{thm:4},
we obtain the following.

\begin{theorem}
\label{thm:6}
(1) $\FROB$ is weakly $\SIGMAP{2}$-complete.
(2) $\FFROB$ is weakly $\SIGMAP{2}$-hard.
\end{theorem}
\begin{proof}
Statements (1) and (2) are proven in almost the same way.
Thus, we only prove statement (2).
By \cite{nijenhuis1979},
$g(A)$ can be computed by the following algorithm 
for a given set $A$ of coprime positive integers $a_1,\cdots,a_n$.
This algorithm is for a single-source shortest path search
for a weighted directed graph $G = (V, E)$ defined from $A$.
The set $V$ of vertices is defined as $[0,a_1-1]$,
and the set $E$ of edges is defined as 
$\{(i,j) \colon i\in V, j = (i+a_l)\mod a_1 \text{\ for some\ } 1 \leq l \leq n\}$.
For every edge $(i,j)$ in $E$,
the distance of $(i,j)$ is defined as $a_l$
such that $j = (i+a_l) \mod a_1$,
where $1 \leq l \leq n$.
Then, the distance $v$ between vertex $0$ and the farthest vertex 
is known to be equal to $g(A)+a_1$~\cite{nijenhuis1979}.
This algorithm runs in time with an order polynomial in $n$ 
but in time with an order exponential in $\log a_1$.
However, if we restrict instances of $\FROB$ to
pairs $E=\{e_1,\cdots,e_m\}$ such that $m = \Omega(e_1)$,
then the above algorithm is a polynomial-time algorithm 
for $n$ and $\sum_{i=1}^n(\lfloor \log e_i \rfloor + 1)$.
It follows that the algorithm of \cite{nijenhuis1979}
is a pseudopolynomial algorithm.
Thus, by Theorem~\ref{thm:4},
$\FFROB$ is a weakly $\SIGMAP{2}$-hard problem.
\end{proof}

Many weakly $\NP$-hard problems
have a fully polynomial-time approximation scheme (FPTAS)~\cite{Garey:1978:SNR:322077.322090}.
Thus, $\FFROB$ might also have an FPTAS.
However, it is known that not all weakly $\NP$-hard problems have an FPTAS~\cite{433837220000101}.
Moreover,
there has not been much attention on approximation for $\SIGMAP{2}$-hard problems.
A result for the approximability of $\FFROB$ may have a significant effect
for computational complexity theory.

\section{Conclusions and Future Work}
\label{sec:conclusion}

In this paper, as a main theorem,
we proved the $\SIGMAP{2}$-completeness of a decision version of the Frobenius problem
under Karp reductions.
This result is an answer for the long-standing open problem proposed 
by Ram\'irez-Alfons\'in (Section A.1 in \cite{ramirezAlfonsin2005diophantine}).
This result provided the first nontrivial upper bound 
and an improved lower bound for the computational complexity
of the original version of the Frobenius problem.
Moreover, as a further improvement trial, 
we proved the membership of $\EFROB$ to $\DIFFP{2}$.
For developing practically fast algorithms,
although our $\SIGMAP{2}$-completeness proof for $\FROB$ showed a negative fact,
we also pointed out a positive aspect that $\FFROB$ is not strongly $\SIGMAP{2}$-hard.
On the other hand,
this paper leaves the following questions open.

\begin{conjecture}
Is $\EFROB$ $\DIFFP{2}$-complete under Karp reductions?
\end{conjecture}

\begin{conjecture}
Is $\FFROB$ $\FUNC\DELTAP{3}$-complete 
or $\mathcal{C}$-complete under Levin reductions,
where $\mathcal{C}$ is some subclass of functions
of $\FUNC\DELTAP{3}$?

\end{conjecture}

Many $\SIGMAP{2}$-complete problems have been known prior to this work~\cite{schaefer2002completeness}. 
However, among those $\SIGMAP{2}$-complete problems,
number theory problems are few.
A computational good property of the Frobenius problem can be derived by 
a future sophisticated result in number theory.
That is, our $\SIGMAP{2}$-completeness proof for $\FROB$ 
may become a trigger for resolving 
some important open questions in theoretical computer science, such as 
$\NP$ vs $\SIGMAP{2}$.
For example, there are the following possibilities.
By developing a proof technique that applies a new characterization 
of the Frobenius number,
we may prove the existence of some instances that
cannot be solved in nondeterministic polynomial time.
Conversely, 
by finding a subproblem of the Frobenius problem,
which can simulate a $\SIGMAP{2}$-complete problem
and can be solved in nondeterministic polynomial time,
we may obtain the result that $\NP = \SIGMAP{2}$.
Naturally, it appears to be highly unlikely,
and discovering such a subproblem is quite difficult.
However, 
since several subproblems that can be efficiently computed
are known~(Sections 3.3 and 3.7 in \cite{ramirezAlfonsin2005diophantine}), 
the existence of such an $\NP$ algorithm could be expected.

\end{document}